\documentclass[3p,12pt]{elsarticle}

\usepackage{graphicx}
\usepackage[space]{grffile}
\usepackage{color}
\usepackage{here}
\usepackage[hang]{caption}
\usepackage{subfig}
\usepackage{empheq}

\usepackage{newtxtext,}
\usepackage{comment}
\usepackage{cases}
\usepackage{longtable}

\usepackage[ruled, vlined]{algorithm2e}
\allowdisplaybreaks
\usepackage{enumitem}
\setlist[itemize]{noitemsep}
\setlist[description]{noitemsep}

\usepackage{amsmath,amssymb,amsthm,amsfonts,bm,mathtools}

\newtheorem{proposition}{Proposition}

\newtheorem{corollary}{Corollary}

\makeatletter
\def\maketag@@@#1{\hbox{\m@th\normalfont\normalsize#1}}
\newcommand{\pushright}[1]{\ifmeasuring@#1\else\omit\hfill$\displaystyle#1$\fi\ignorespaces}
\newcommand{\pushleft}[1]{\ifmeasuring@#1\else\omit$\displaystyle#1$\hfill\fi\ignorespaces}
\makeatother

\makeatletter
	\let\start@align@nopar\start@align
	\let\start@gather@nopar\start@gather
	\let\start@multline@nopar\start@multline
	\let\start@numcases@nopar\start@numcases
	\long\def\start@align{\par\vspace{-1em}\start@align@nopar}
	\long\def\start@gather{\par\vspace{-1em}\start@gather@nopar}
	\long\def\start@multline{\par\vspace{-1em}\start@multline@nopar}
	\long\def\start@numcases{\par\vspace{-1em}\start@numcases@nopar}
\makeatother

\journal{Elsevier}
\biboptions{authoryear}

\newcommand{\tabref}[1]{Tab.~\ref{#1}}
\newcommand{\equref}[1]{Eq.~(\ref{#1})}
\newcommand{\figref}[1]{Fig.~\ref{#1}}

\def\del{\partial}
\newcommand{\maxm}{\mathop{\text{max}}\limits}

\def\({\left(}
\def\){\right)}
\def\<{\left[}
\def\>{\right]}

\def\aa{{\bf a}}
\def\AA{{\bf A}}
\def\NN{{\bf N}}
\def\PP{{\bf P}}

\def\mf{{m_+}}
\def\ml{{m_-}}

\def\kj{\kappa}

\def\F{{\rm F}}
\def\C{{\rm C}}
\def\tthh{\bm{\theta}}

\begin{document}
\begin{frontmatter}
	\title{Fundamental diagram estimation by using trajectories of probe vehicles}

	\author[ut]{Toru Seo\corref{mycorrespondingauthor}}
	\cortext[mycorrespondingauthor]{Corresponding author}
	\ead{seo@civil.t.u-tokyo.ac.jp}
	\author[tokyotech]{Yutaka Kawasaki}
	\author[utcsis]{Takahiko Kusakabe}
	\ead{t.kusakabe@csis.u-tokyo.ac.jp}
	\author[tokyotech]{Yasuo Asakura}
	\ead{asakura@plan.cv.titech.ac.jp}
	\address[ut]{The University of Tokyo, 7-3-1 Hongo, Bunkyo-ku, Tokyo 113-8656, Japan}
	\address[tokyotech]{Tokyo Institute of Technology, 2-12-1-M1-20, O-okayama, Meguro, Tokyo 152-8552, Japan}
	%\address[um]{University of Michigan, 2350 Hayward St., Ann Arbor, MI 48109, The United States}
	\address[utcsis]{Center for Spatial Information Science, the University of Tokyo, 5-1-5 Kashiwanoha, Kashiwa-shi, Chiba 277-8568, Japan}
	\begin{abstract}
		The fundamental diagram (FD), also known as the flow--density relation, is one of the most fundamental concepts in the traffic flow theory.
		Conventionally, FDs are estimated by using data collected by detectors.
		However, detectors' installation sites are generally limited due to their high cost, making practical implementation of traffic flow theoretical works difficult.
		On the other hand, probe vehicles can collect spatially continuous data from wide-ranging area, and thus they can be useful sensors for large-scale traffic management.
		In this study, a novel framework of FD estimation by using probe vehicle data is developed.
		It determines FD parameters based on trajectories of randomly sampled vehicles and a given jam density that is easily inferred by other data sources.
		A computational algorithm for estimating a triangular FD based on actual, potentially noisy traffic data obtained by multiple probe vehicles is developed.
		The algorithm was empirically validated by using real-world probe vehicle data on a highway.
		The results suggest that the algorithm accurately and robustly estimates the FD parameters.
	\end{abstract}
\end{frontmatter}

\section{Introduction}

The {\it fundamental diagram (FD)}, also known as the flow--density relation, is literally one of the most fundamental concepts in the traffic flow theory.
An FD describes the relation between flow and density in {\it steady traffic} (sometimes referred as equilibrium or stationary traffic), that is, {\it traffic in which all the vehicles exhibit the same constant speed and spacing} \citep{daganzo1997book, Treiber2013traffic}.
In theory, an FD itself contains useful information on traffic features, such as the value of free-flow speed and flow capacity, and distinction between free-flow and congested regimes.
Observational studies also indicated a clear relation between flow and density in actual traffic at near-steady states \citep{Cassidy1998bivariate, Yan2018stationary}.
Additionally, macroscopic traffic flow dynamics are modeled by combining an FD and other principles---the most widely known example is the Lighthill--Whitham--Richards (LWR) model \citep{Lighthill1955LWR, Richards1956LWR}.
Furthermore, FDs explain microscopic vehicle behavior to a certain extent \citep{newell2002carfollowing, Duret2008fd, Jabari2014fd}.
Thus, FDs are utilized by a wide variety of academic and practical purposes in traffic science and engineering fields such as traffic flow modeling/analysis/simulation \citep{Daganzo1994ctm, daganzo1997book, daganzo2006kinematic, Treiber2013traffic}, traffic control \citep{Papageorgiou2003control}, traffic state estimation \citep{Seo2017review}, and dynamic traffic assignment in network \citep{Szeto2006dta}.

In order to estimate a parametric FD in actual traffic, it is necessary to collect data from traffic, assume the functional form of its FD, and estimate the FD parameters by fitting the curve to the data.
Typically, FDs are estimated by using roadside sensors (e.g., cameras and detectors) from the era of \citet{Greenshields1935fd}.
This is a straightforward method because usual roadside sensors measure traffic count and occupancy that are closely related to flow and density, respectively, at their location.

The limitation of the roadside sensor-based FD estimation is evident: it is unable to estimate FDs where sensors are not installed.
Furthermore, because the installation and operational costs of roadside sensors are generally high, it is not practically possible to deploy detectors to everywhere.
Thus, it is not always possible to determine FDs on every roads, especially arterial roads and highways in developing countries.
It is also difficult to identify the exact locations and characteristics of bottlenecks even on freeways with some sensors.
This limitation of detectors poses substantial difficulties for practical implementations of the aforementioned traffic science and engineering works that require pre-determined FDs.

{\it Probe vehicles}\footnotemark{} \citep{sanwal1995probe, Zito1995gps} received increasing attention currently, because of the rise of mobile technology and connected vehicles \citep{Herrera2010probe, Shladover2017autonomous}.
\footnotetext{%{}
	In this study, the term ``probe vehicle'' denotes a vehicle that continuously measures and reports its position and time (i.e., spatiotemporal trajectory) by using global navigation satellite systems, mobile phones, connected vehicle technologies, or other similar on-vehicle systems.
}%
The notable feature of the probe vehicles is that they collect spatially continuous data from wide-ranging area.
Furthermore, the data can be collected during probe vehicles' daily travel, making them a cost-effective traffic data collection tool.
These are remarkable advantages compared with conventional detectors, which can collect data at limited discrete locations \citep{Herrera2010traffic, Jenelius2015probe, Zheng2017intersection}.
If FDs can be estimated by probe vehicle data by leveraging their advantages, it will enable us to estimate FDs in virtually everywhere.
It will further enable us to implement various traffic science and engineering works to large-scale road networks.
However, there is a paucity of research on FD estimation by using probe vehicle data.
Thus, systematic and computational approaches for FD estimation are desirable given the high availability of probe vehicle data.

The aim of this study is to establish methodology to estimate FDs by using probe vehicle data.
In order to enable the estimation, we allow the methodology to rely on minimum exogenous assumptions such as FD's functional form and value of a parameter of the FD (i.e., the jam density that is inferred from external knowledge).
Moreover, this study presents a computable algorithm for FD estimation based on data collected by multiple probe vehicles that are driving potentially noisy real-world traffic.
The algorithm was applied to real-world connected vehicle data to investigate the empirical validity of the proposed methodology.

The rest of the paper is organized as follows.
In Section \ref{sec_rev}, the literature on FD estimation is reviewed and originality of this study is highlighted.
In Section \ref{sec_theory}, the basic idea of FD estimation using probe vehicle data is discussed, and the identifiability of a triangular FD under idealized conditions is theoretically clarified.
Then, in Section \ref{sec_method}, an computational algorithm for estimating a triangular FD by using actual, noisy traffic data is developed based on the basic idea.
The method statistically estimates the most probable FD using maximum likelihood estimation based on data collected by multiple probe vehicles.
In Section \ref{sec_validation}, the empirical performance of the proposed algorithm is verified by applying it to real-world probe vehicle data.
Section \ref{sec_conclusion} concludes this paper.
The frequently used abbreviations and notations are listed in \ref{sec_notation}.

\section{Literature review on FD estimation}\label{sec_rev}

FD estimation by using detectors is extensively studied since 1935.
Many approaches use traffic data aggregated for a short time period (e.g., traffic count for 30~sec--5~min periods) as inputs of estimation and extract statistical relation from them \citep[see][and references therein]{Treiber2013traffic, Qu2017fd, Knoop2017fd}.
Other approaches use disaggregated data that distinguish each individual vehicle in measurement area \citep[e.g.,][]{Duret2008fd, Chiabaut2009fd, Coifman2014fd, Coifman2015fd, Jin2015fd}.
The latter approaches enable the precise identification of steady traffic states.
These methodologies rely on the fact that flow and density or trajectories of all the vehicles are easily obtained by detectors.
However, this is not the case for typical probe vehicles: only trajectories of randomly sampled vehicles are collected by probe vehicles.
Thus, these methodologies are not applicable for FD estimation by using probe vehicle data.

There is a paucity of research on FD estimation by using probe vehicle data.
A few probe vehicle-based traffic state estimation methods simultaneously estimate FDs and traffic state \citep[e.g.,][]{Tamiya2002probe, Seo2015ieee, sun2017tse}.
However, they rely on additional data sources, such as detector measurement or vehicle spacing in addition to probe vehicle's positioning.
Appendix in \citet{Herrera2010probe} mentioned an idea for manual inference of an FD by using probe vehicle data under special conditions. 
However, it was not formulated and not computable.
Their idea involves manually detecting a shockwave between a saturated free-flowing traffic and congested traffic, deriving the backward wave speed of a triangular FD by manually measuring the speed of the shockwave, and subsequently calculating the other FD parameters by using a given jam density.
However, detection of a shockwave, detection of saturated traffic, and identification of the speed of shockwave are not trivial tasks in actual traffic.

\section{Identification of FD by using probe vehicles under idealized conditions}\label{sec_theory}

In this section, we mathematically investigate identifiability of FDs by using probe vehicle data under idealized conditions, such as homogeneous traffic described by the LWR model and accurate measurement.

\subsection{Preliminaries}\label{sec_prel}

\subsubsection{Assumptions}\label{sec_assumption}

We assume that traffic dynamics are described by the LWR model \citep{Lighthill1955LWR, Richards1956LWR} with a homogeneous FD and first-in first-out (FIFO) condition without a source or sink.
It is expressed as
\begin{align}
	&\frac{\del k(t,x)}{\del t} + \frac{\del q(t,x)}{\del x} = 0,	\\
	&q(t,x) = Q(k(t,x); \tthh),	\label{fd}
\end{align}
where $q(t,x)$ and $k(t,x)$ denote flow and density, respectively, at time--space point $(t,x)$,
$Q$ denotes a flow--density FD, and
$\tthh$ denotes a vector representing the FD parameters.

The functional form of an FD $Q(k; \tthh)$ is assumed as to be {\it triangular} \citep{Newell1993kinematic1}, which is expressed as
\begin{subequations}
\begin{numcases}{Q(k;u,w,\kj)=}
	uk	&	if $k \leq \frac{w\kj}{u+w}$	\label{freeflow}\\
	w(\kj-k) & otherwise,		\label{congested}
\end{numcases}
\end{subequations}
where $u$ denotes the free-flow speed, $w$ denotes the backward wave speed, and $\kappa$ denotes the jam density.
Thus, $\tthh$ consists of $u$, $w$, and $\kappa$.
Equation~\eqref{freeflow} represents a {\it free-flowing regime}, whereas Eq.~\eqref{congested} represents a {\it congested regime}.
The triangular FD is employed owing to its simplicity, good agreement with actual phenomena, and popularity in various applications \citep[e.g.,][]{daganzo2006kinematic}.
The triangular FD is employed owing to its simplicity, good physical interpretation of parameters, and popularity in various applications \citep[e.g.,][]{daganzo2006kinematic}.
Note that it is easy to generalize the proposed methodology to incorporate single-valued, continuous, piecewise/fully differential FDs; see \ref{apx_general}.

A probe vehicle dataset is assumed to contain continuous trajectories of randomly sampled vehicles without any errors.
It is represented as $\{X(t,n)~|~\forall n \in \PP\}$, where $X(t,n)$ denotes the position that vehicle $n$ exists at time $t$, and $\PP$ denotes a set of all the probe vehicles.
The penetration rate of probe vehicles is unknown.

\subsubsection{Definitions}\label{sec_definition}

Let {\it traffic state} in a time--space region be defined by Edie's generalized definition \citep{Edie1963generalized}, namely,
\begin{subequations}
\begin{align}
	q(\AA) &= \frac{\sum_{n \in \NN(\AA)} d_n(\AA)}{|\AA|},	\\
	k(\AA) &= \frac{\sum_{n \in \NN(\AA)} t_n(\AA)}{|\AA|},	\\
	v(\AA) &= \frac{\sum_{n \in \NN(\AA)} d_n(\AA)}{\sum_{n \in \NN(\AA)} t_n(\AA)},
\end{align}%
\end{subequations}%
where $\AA$ denotes a time--space region, and
$q(\AA)$, $k(\AA)$, and $v(\AA)$ represent flow, density, and speed, respectively, in the time--space region $\AA$,
$\NN(\AA)$ denotes a set of all the vehicle in $\AA$,
$d_n(\AA)$ and $t_n(\AA)$ denote distance traveled and time spent, respectively, by vehicle $n$ in $\AA$, and 
$|\AA|$ denotes the area of $\AA$.
If traffic in $\AA$ is steady,
\begin{align}
	q(\AA) = Q(k(\AA); \tthh)	\label{fda}
\end{align}
is satisfied by the definition of the LWR model \eqref{fd}.

Let us consider a {\it probe pair}, that is, a pair of arbitrary two probe vehicles.
A pair is denoted by $m$, and the preceding probe vehicle in the pair $m$ is denoted by $\ml$ whereas the following probe vehicle is denoted by $\mf$.
Note that arbitrary and unknown number of vehicles may exist between a probe pair.

Let $(q_m(\AA),k_m(\AA))$ be {\it probe traffic state (PTS)} of probe pair $m$ defined as
\begin{subequations}
\begin{align}
	q_m(\AA) &= \frac{d_\mf(\AA)}{|\AA|},	\\
	k_m(\AA) &= \frac{t_\mf(\AA)}{|\AA|},	
\end{align}	\label{pts_def}%
\end{subequations}%
following Edie's definition.
Similarly, probe-speed is also defined as
\begin{align}
	v_m(\AA) &= \frac{d_\mf(\AA)}{t_\mf(\AA)}.
\end{align}
The PTS denotes flow and density of probe vehicle $\mf$ only in time--space region $\AA$.
If (i) traffic in $\AA$ is steady, (ii) distance traveled by each vehicle in $\AA$ is equal to the others, and (iii) probe vehicle $\mf$ exists in $\AA$, then the following relations are satisfied:
\begin{subequations}
\begin{align}
	q_m(\AA) &= \frac{q(\AA)}{|\NN(\AA)|},	\label{pq}\\
	k_m(\AA) &= \frac{k(\AA)}{|\NN(\AA)|},	\label{pk}\\
	v_m(\AA) &= v(\AA),
\end{align}%
\end{subequations}%
where $|\NN(\AA)|$ denotes the size of $\NN(\AA)$ (i.e., number of the vehicles in $\AA$).

Let $c_m-1$ denote the number of vehicles between probe pair $m$, and $\aa_m$ denote a time--space region whose boundary is on trajectories of $\mf$ and $\ml$ (the region includes $\mf$ but not $\ml$).
Since source/sink and FIFO-violations are absent, $c_m$ is constant and 
\begin{align}
	c_m = |\NN(\aa_m)|
\end{align}
always holds irrespective of shape, time, and location of $\aa_m$.
If (i) traffic in $\aa_m$ is steady and (ii) distance traveled by each vehicle in $\aa_m$ is equal to the others, then
\begin{align}
	q_m(\aa_m) = \frac{Q(c_m k_m(\aa_m); \tthh)}{c_m}	\label{fd-pfd}
\end{align}
is satisfied from Eqs.~\eqref{fda}, \eqref{pq}, and \eqref{pk}.
Hereafter, a PTS of pair $m$ under the aforementioned conditions (i) and (ii) is simply referred to as {\it steady PTS}.
Equation~\eqref{fd-pfd} means that steady PTSs follow a relation that is similar (in the geometric sense) to the FD with a scale of $1/c_m$ on a flow--density plane.
Hereafter, this relation is referred to as {\it probe fundamental diagram (PFD)} of $m$.
Note that the PFD of probe pair $m$ is also expressed
\begin{align}
	q_m(\aa_m) = Q(k_m(\aa_m); \tthh_m).		\label{pfd}
\end{align}
where $\tthh_m$ denotes a pair-specific parameter vector, because the FD and PFD are similar.
Because of the similarity, $\tthh_m$ consists of free-flow speed of the PFD, $u$, backward wave speed of the PFD, $w$, and jam density of the PFD, $\kappa/c_m$.

In order to ensure that $\aa_m$ easily satisfies the aforementioned conditions (i) and (ii), we specify $\aa_m$ as follows:
\begin{align}
	\aa_{mi} = \{(t,x)~|~	&X(t,\mf) \leq x < X(t,\ml) \textrm{ and }\notag\\ 
							&-\phi(t-\tau_i)+X(\tau_i, m_+) \leq x \leq -\phi(t-\tau_i-\Delta t)+X(\tau_i+\Delta t, m_+) \},	\label{adef}
\end{align}
where $i$ denotes an index number, $\tau_i$ denotes a reference time for $\aa_{mi}$, and $\phi$ and $\Delta t$ are the given parameters.
An arbitrary value is acceptable for $\tau_i$ and $\phi$, and an arbitrary non-zero value is acceptable for $\Delta t$ if the corresponding trajectory $X(t,n)$ exists.
The shape of $\aa_m$ is illustrated in \figref{shape_a}.
The specification \eqref{adef} guarantees that if traffic in $\aa_m$ is steady, then distance traveled by each vehicle in $\aa_m$ is equal to the others, which is $d_{m_+}(\aa_{mi})$.
This feature ensures that the following discussion is concise without sacrificing significant generality.

\begin{figure}[!htb]
	\centering
	\includegraphics[clip, width=0.6\hsize]{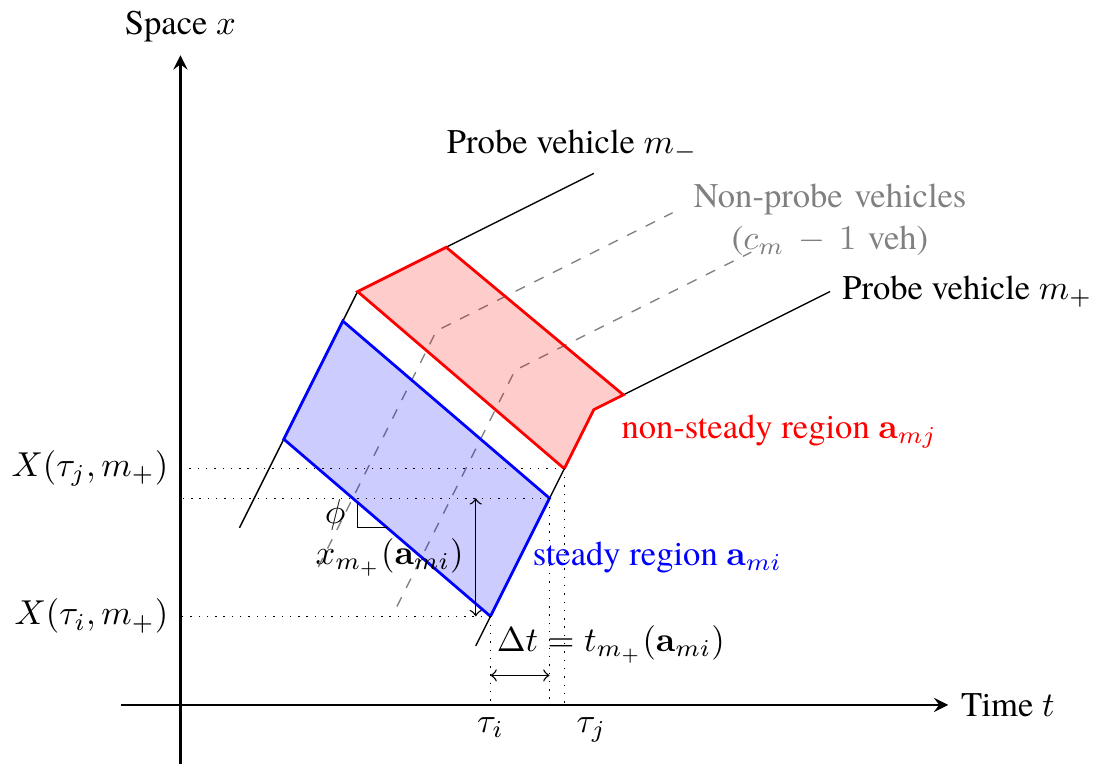}
	\caption{Shape of the time--space region $\aa_{m}$.}
	\label{shape_a}
\end{figure}

An FD is said to be {\it identifiable} by using a given dataset (a set of traffic states or a set of PTSs) if the value of its parameter vector $\tthh$ is determined based on the dataset.
Similarly, a PFD of certain probe pair $m$ is said to be identifiable by using the given dataset (a set of PTSs of $m$) if the value of its parameter vector $\tthh_m$ is determined based on the dataset.

For the readers' convenience, important abbreviations and notations used throughout this paper are summarized in \ref{sec_notation}.

\subsection{Identifiability of FD based on probe vehicle data}

It is mathematically obvious that following proposition holds true:
\begin{proposition}\label{pro_fd_identify_triangle}
	A triangular FD is identifiable if (i) one or more steady traffic state data from the free-flowing regime are available and (ii) two or more different steady traffic state data from the congested regime are available.
\end{proposition}\noindent
Given the relation between a FD and a corresponding PFD of a probe pair shown in Eq.~\eqref{fd-pfd}, following corollary can be obtained from Proposition \ref{pro_fd_identify_triangle}:
\begin{corollary}\label{coro_pfd_identify_triangle}
	A triangular PFD is identifiable if (i) one or more steady PTS data from the free-flowing regime are available for a particular probe pair and (ii) two or more different steady traffic state data from the congested regime are available for the same probe pair.
\end{corollary}\noindent
This means that the value of the free-flow speed and the backward wave speed can be determined by appropriate probe vehicle data.

According to Eq.~\eqref{fd-pfd} again, an PFD is similar (in the geometric sense) to a corresponding FD with a scale of $1/c_m$ on a flow--density plane.
Therefore, once an PFD is identified and the value of the jam density is given, the corresponding FD can be identified.
By combining this fact and Corollary \ref{pro_fd_identify_triangle}, following corollary can be obtained:
\begin{corollary}\label{coro_fd_identify_triangle}
	A triangular FD is identifiable if (i) one or more steady PTS data from the free-flowing regime are available for a probe pair, (ii) two or more different steady PTS data from the congested regime are available for the probe pair, and (iii) the value of the jam density is known.
\end{corollary}\noindent
Notice that it is not necessary for an analyst to know whether a specific PTS belongs to free-flowing or congested regime beforehand.
The belongingness is determined as a result of the identification of the FD.

\subsection{Geometric interpretation of FD identification based on probe vehicle data}

Geometric interpretation of the triangular FD identification based on Corollary \ref{coro_fd_identify_triangle} is explained using \figref{concept}.
\figref{concept} shows traffic flow consisting of free-flow traffic and congested traffic as a time--space diagram (top) and a flow--density plane (bottom).
In the time--space diagram, probe vehicles consisting of pair $m$ are denoted by solid lines, non-probe vehicles are denoted by dashed gray lines, and a shockwave is denoted by an arrow.
In the flow--density plane, dots indicate traffic states or PTSs, solid lines indicate transition of traffic states; the dashed line denotes the FD, and the thick dashed line denotes the PFD.
Notice that the FD and the PFD are similar (in the geometric sense) with scale of $1/c_m$ as in \equref{fd-pfd}.
The blue area (top) and dot (bottom) indicate a steady region with free-flow traffic, the green ones indicate non-steady regions, and the red ones indicate steady region with congested traffic.

\begin{figure}[!htb]
	\centering
	\includegraphics[clip, width=0.86\hsize]{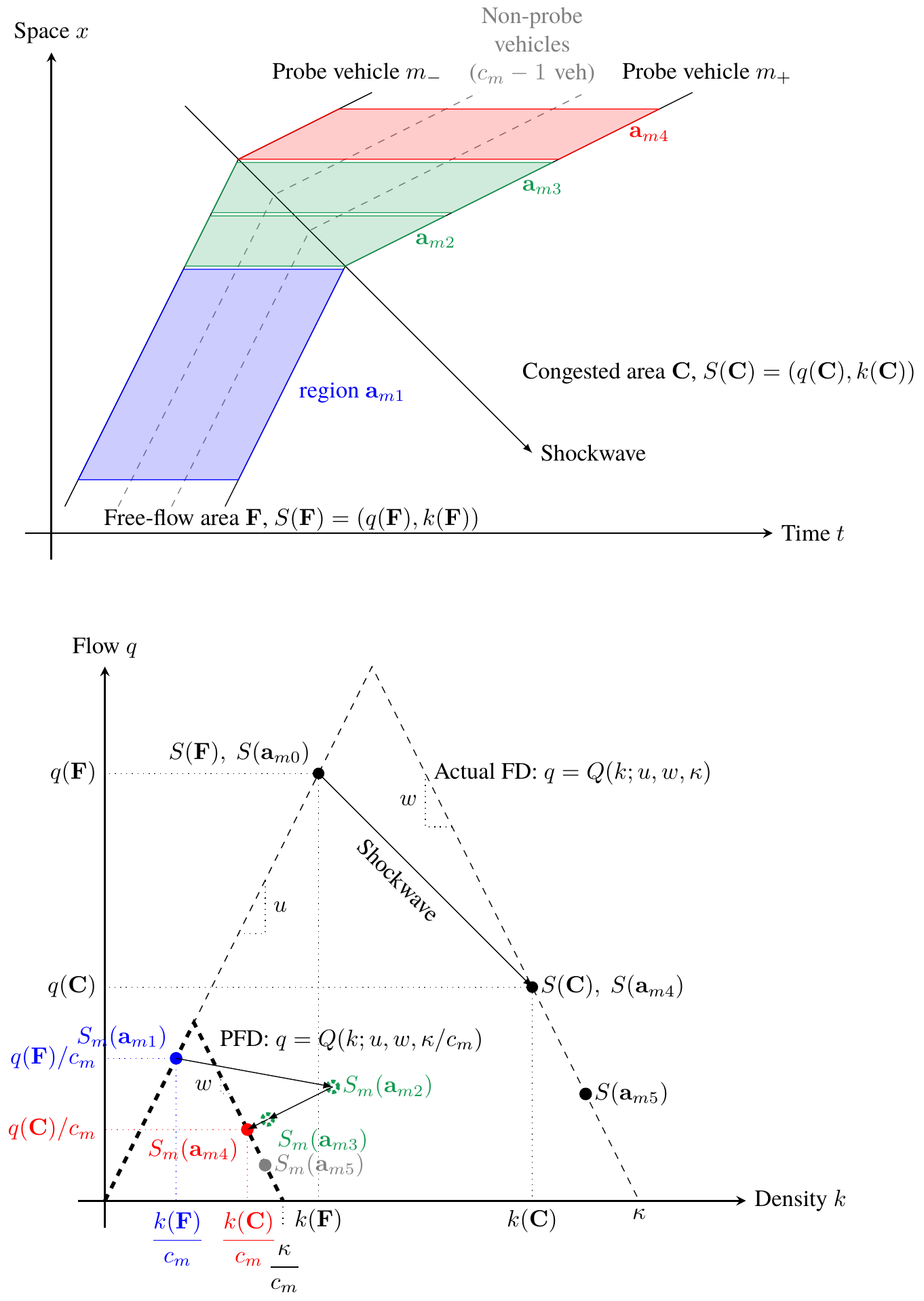}
	\caption{Illustration of triangular FD identification. Top: Probe pair and PTS in time--space diagram. Bottom: FD, PFD, and PTS on flow--density plane.}
	\label{concept}
\end{figure}

As explained, traffic state $S_m(\aa_{mi})$ follows the FD if region $\aa_{mi}$ is steady.
Similarly, PTS $S_m(\aa_{mi})$ is equal to $S_m(\aa_{mi})/c_m$ and follows the PFD if region $\aa_{mi}$ is steady.
In \figref{concept}, regions $\aa_{m1}$ and $\aa_{m4}$ are steady; and thus their PTS are on the PFD.
Conversely, regions $\aa_{m2}$ and $\aa_{m3}$ are non-steady; and thus their PTS are not on the PFD.
Steadiness of region $\aa_{mi}$ can be inferred by using trajectories of probe vehicles.
The necessary conditions for the steadiness are that the two probe vehicles' speeds are time-invariant, and they are equal to each other.

Suppose that the probe vehicles $m_-$ and $m_+$ traveled to another congested region, say $\aa_{m5}$ (the gray dot in \figref{concept}).
In this case, the PFD is uniquely determined by selecting values of $u$, $w$, and $\kappa/c_m$ such that all the points $S_m(\aa_{m1})$, $S(\aa_{m4})$, and $S(\aa_{m5})$ are on the PFD $q=Q(k;u,w,\kj/c_m)$.
In other words, a triangle that satisfy following conditions can be determined uniquely: its vertex is at point $(0,0)$, one of its edges is on line $q=0$, and its another two edges pass the blue point, red point, and gray point.
Thus, the FD is identified from the probe vehicle data and the jam density.

\subsection{Discussion}\label{sec_discussion}

The proposed method identifies all the parameters of an FD with the exception of its jam density by using probe vehicle trajectory data.
{The value of the jam density is required to be known prior to the FD identification, as it relates an FD and a PFD.}
There are several methods to obtain the value: common knowledge and remote sensing would be especially useful at this moment.

The jam density is considerably time- and space-independent variable compared to the other FD parameters such as capacity.
Thus, it is reasonable to set the jam density based on a value that is observed at nearby locations.
Alternatively, remote sensing (e.g., image recognition) from satellites is used to measure the jam density directly.
Although remote sensing does not provide us time-continuous traffic state variables, such as flow and speed (because the time interval between two measurements is extremely long, few hours to few days at least), it measures space-continuous density accurately \citep{mccord2003satellite}.
Thus, the jam density could be inferred as an upper limit of such measured density.

Besides, the proposed method identifies the values of the free-flow speed and the backward wave speed even if the value of the jam density is unknown as explained by Corollary \ref{coro_pfd_identify_triangle}.
Since these two variables are important to characterize traffic flow, this property would be useful for various purposes such as traffic speed reconstruction and travel time estimation \citep{treiber2011tse}.

A limitation of the proposed method is that a homogeneous FD is assumed.
It means that the method in itself is unable to consider a bottleneck.
This issue can be resolved by extending the method: for example, it is possible to estimate section-dependent FDs by applying the proposed method to each section iteratively.
Another limitation is that the proposed method does not consider phenomena related with multi-valued and/or non-continuous FDs such as a capacity drop.

\section{Estimation algorithm for actual traffic}\label{sec_method}

Actual traffic data are not exactly consistent with the theoretical LWR model, and thus it is not possible to directly apply the methodology described in Section \ref{sec_theory} to actual data---all the data will be considered as non-steady, and {PTSs reported by different probe pairs may not derive the same PFD and FD.}
The inconsistency is mainly due to two sources.
First, various traffic phenomena, such as the existence of non-equilibrium flow, bounded acceleration, heterogeneity among vehicles, FIFO violation in multi-lane sections, and the resulting stochasticity of macroscopic features, are not considered by the LWR model \citep{Jabari2014fd, Jin2015fd, Coifman2015fd, Qu2017fd, Jin2018bounded}.
Second, data always include measurement noise.

{In order to account for this challenge, this section presents a statistical estimation algorithm for a triangular FD from actual traffic data collected from multiple probe pairs.
Specifically, a method of extracting near-steady traffic state data from probe vehicle data is presented in Section \ref{sec_stat}.
Then, an estimation method for free-flow speed and backward wave speed based on near-steady traffic state data collected by multiple probe pairs, which may exhibit noisy and heterogeneous behavior, is presented in Sections \ref{sec_back} and \ref{sec_em}.}

In order to simplify the mathematical representation, the triangular FD is characterized by free-flow speed, $u$, backward wave speed, $w$, and, $y$-intercept of the congestion part of FD, $\alpha$ as illustrated in \figref{triangularfd}.
Note that $\alpha \equiv \kappa w$ holds, and thus, the characterization does not contradict the assumptions in Section \ref{sec_prel}.

\begin{figure}[htb]
	\centering
	\includegraphics[clip, width=0.5\hsize]{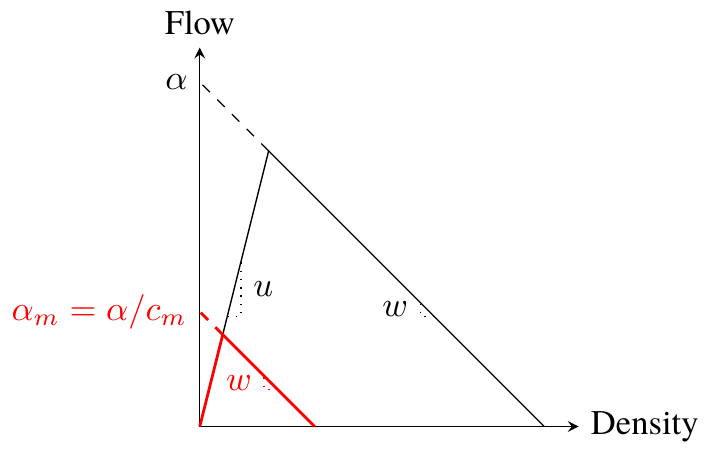}
	\caption{Triangular FD with its three parameters $u$, $w$, and $\alpha$. Black curve denotes FD, and red curve denotes the corresponding PFD of pair $m$.}
	\label{triangularfd}
\end{figure}

\subsection{Extraction of near-steady PTSs}\label{sec_stat}

Suppose that multiple PTSs $S_m(\aa_{mi})$ for multiple $m$ and $i$ were calculated from probe vehicle data based on the definition \eqref{pts_def} and \eqref{adef}.
This set of PTSs is referred to as {\it raw PTS set}.

It is necessary to extract PTSs that are near-steady, meaning that they nearly satisfy the conditions in Corollary \ref{coro_fd_identify_triangle} from the raw PTS set.
To do this, two filtering methods are employed.
The first method is based on the linearity of the trajectories of two probe vehicles in pair $m$.
Specifically, $S_m(\aa_{mi})$ is considered as near-steady if trajectories of pair $m$ in region $\aa_{mi}$ satisfy a condition based on the coefficient of variation:
\begin{align}
	\frac{\sigma_m(\aa_{mi})}{\bar{v}_m(\aa_{mi})} \leq \theta_\text{steady}		\label{statidentify}
\end{align}
where $\sigma_m(\aa_{mi})$ and $\bar{v}_m(\aa_{mi})$ denote the standard deviation and mean, respectively, of the instantaneous speed of probe vehicles in pair $m$ in region $\aa_{mi}$, and
$\theta_\text{steady}$ denotes a given threshold.

In the second method, $\{S_m(\aa_{mi})~|~\forall i\}$ of each $m$ is verified as to whether it is likely to satisfy the conditions in Corollary \ref{coro_fd_identify_triangle}.
Data from the probe pair $m$ are discarded if $\{S_m(\aa_{mi})~|~\forall i\}$ included no $S_m(\aa_{mi})$ with speed lower than $u_{\min}$ or no $S_m(\aa_{mi})$ with speed faster than $u_{\min}$, where $u_{\min}$ denotes a given lower bound for free-flow speed.\footnotemark{}
\footnotetext{%{}
	Note that although a lower bound for free-flow speed $u_{\min}$ is used in this procedure, it is not necessary to know the precise value of free-flow speed before estimation. 
	This point is discussed later.
}
This corresponds to conditions (i) and (ii) in Corollary \ref{coro_fd_identify_triangle}. 
Additionally, data from probe pair $m$ are discarded if correlation between probe-flow and probe-density in $\{S_m(\aa_{mi})~|~\forall i, v(\aa_{mi}) \leq u_{\min}\}$ exceeds $\theta_{\text{corr}}$, or standard deviation in $\{\bar{v}_m(\aa_{mi})~|~\forall i, v(\aa_{mi}) \leq u_{\min}\}$ is smaller than $\theta_\text{std}$, where $\theta_{\text{corr}}$ and $\theta_\text{std}$ denotes given thresholds.
This corresponds to condition (ii) in Corollary \ref{coro_fd_identify_triangle}.

Hereafter, a set of PTS data obtained by the above procedure is referred to as {\it near-steady PTS set}.
A set of probe pairs that collected elements in the near-steady PTS set is denoted as $M$, and a near-steady PTS set obtained by pair $m$ is denoted as $I_m$.

\subsection{Backward wave speed estimation}\label{sec_back}

Backward wave speed is first estimated separately from the free-flow speed.
This is because it is easy to determine whether {\it a PTS datum definitely belongs to congested regimes} based on the lower bound of the free-flow speed $u_{\min}$.

Specifically, the backward wave speed $w$ is estimated by performing a linear regression on $(k(\aa_{mi}), q(\aa_{mi}))$ for each $i \in I_m$, $m \in M$, and $v_m(\aa_{mi}) \leq u_{\min}$; and the final estimate is given as the mean of all the regression results.
Condition (ii) in Corollary \ref{coro_fd_identify_triangle} is essential to estimate $w$ appropriately by the method.

\subsection{Free-flow speed estimation by using EM algorithm}\label{sec_em}

The free-flow speed $u$ is subsequently estimated.
To do this, an iterative algorithm categorized as expectation--maximization (EM) algorithm \citep{dempster1977em} is constructed.
In this algorithm, free-flow $u$ is estimated in conjunction with the $y$-intercept value $\alpha_m$ and other variables, such as standard deviation of free-flowing and congested part of an FD, based on the near-steady PTS set.
	
\subsubsection{Likelihood function}

Assume that a probe-flow of a steady PTS belonging to regime $s$ (free-flowing, $s=\F$, or congested, $s=\C$) follows a normal distribution in which the mean is the true PFD and standard deviation is $r_m\sigma_{s}$, where $\sigma_{s}$ denotes a standard deviation of traffic state from an FD, $r_m$ denotes a pair-specific normalizing parameter defined as $\bar{q}_m/\bar{q}$, $\bar{q}_m$ denotes the mean of probe-flow of pair $m$, and $\bar{q}$ denotes the mean of probe-flow of all of the probe pairs.
Thus, likelihood of regime $s$ of PTS datum $i$ of pair $m$ is expressed as
\begin{align}
	l(m,i,s) = \left\{\begin{array}{ll}
		f\(q_{mi} - u k_{mi}, r_m\sigma_s\), & \text{if } s = \F,		\\
		f\(q_{mi} - (\alpha_m -w k_{mi}), r_m\sigma_s\), & \text{if } s = \C,		\\
	\end{array}\right.
\end{align}
where $f(x,\sigma)$ denotes a value of the probability distribution function of a normal distribution with mean $0$ and standard deviation $\sigma$ at $x$.
Note that variables $u$, $\alpha_m$, $s$, and $\sigma_s$ are unknown, while the rest are known.

The likelihood of the mixture distribution for all the near-steady PTS set is given by
\begin{align}
	L\(q, k | \pi, u, \{\alpha_m~|~\forall m\}, \{\sigma_s~|~\forall s\}\) = \prod_{m \in M} \prod_{i \in I_m} \sum_{s \in \{\F, \C\}} \pi_s l(m,i,s),		\label{trueL}
\end{align}
where $\pi_s$ denotes the contribution of state $s$ in the mixed distribution \citep[c.f.,][]{Bishop2006en}.
The intuitive meaning of $L$ is a likelihood of observing PTS $(q_{mi}, k_{mi})$ set when parameters\footnotemark{} $\pi, u, w, \{\alpha_m\}, \{\sigma_s\}$ are given.
\footnotetext{%{}
	Hereafter, $\{\alpha_m~|~\forall m\}$ and $\{\sigma_s~|~\forall s\}$ are abbreviated as $\{\alpha_m\}$ and $\{\sigma_s\}$, respectively, for the purposes of simplicity.
}%

Our problem is to estimate the latent parameters $\pi, u, \{\alpha_m\}, \{\sigma_s\}$ from the observed near-steady PTS $(q_{mi}, k_{mi})$ set.
This can be solved by applying an EM algorithm.
If the complete dataset \citep[c.f.,][]{Bishop2006en} is available (i.e., latent variable $z_{mis}$ that is 1 if $S_{mi}$ belongs to regime $s$, 0 otherwise is given), then Eq.~\eqref{trueL} is transformed
\begin{align}
	L(q, k | \pi, u, \{\alpha_m\}, \{\sigma_s\}, z) = \prod_{m \in M} \prod_{i \in I_m} \prod_{s \in \{\F, \C\}} \pi_s^{z_{mis}} l(m,i,s)^{z_{mis}}.		\label{quasiL}
\end{align}
By taking its expectation of logarithm, one obtains\footnotemark{}
\footnotetext{%{}
	Hereafter, $\sum_{m \in M}, \sum_{i \in I_m}$, and $\sum_{s \in \{\F, \C\}}$ are abbreviated as $\sum_m, \sum_i$, and $\sum_s$, respectively, for the purposes of simplicity.
}%
\begin{align}
	E[LL(q, k | \pi, u, \{\alpha_m\}, \{\sigma_s\}, \{p_{mi}(s)\})] = \sum_m \sum_i \sum_s p_{mi}(s) \(\ln \pi_s + \ln l(m,i,s)\),	 	\label{ELL}
\end{align}
where $p_{mi}(s)$ denotes an expectation of $z_{mis}$.
The EM algorithm determines the value of latent parameters $\pi, u, w, \alpha, \sigma$, and $p$ by maximizing Eq.~\eqref{ELL} with an iterative procedure termed as E-steps and M-steps.

\subsubsection{E-step}

In the E-step, $p_{mi}(s)$ is updated under given $\pi_s,u,w,\alpha_m,\sigma$.
It is expressed as
\begin{align}
	p_{mi}(s) = \frac{\pi_s l(m,i,s)}{\sum_{r \in \{\F,\C\}} \pi_r l(m,i,r)},	\qquad \forall s \in \{\F, \C\},~\forall i \in I_m,~\forall m \in M.
\end{align}

\subsubsection{M-Step}

\def\alphamin{0}%\alpha_{\min}

In the M-step, $\pi_s,u,\{\alpha_m\},\{\sigma_s\}$ is updated such that the likelihood \eqref{ELL} is maximized under given $p_{mi}(s)$.
Specifically, it is expressed as
\begin{subequations}
\begin{align}
	\maxm_{\pi, u, \{\alpha_m\}, \{\sigma_s\}}~& E[LL(q, k | \pi, u, \{\alpha_m\}, \{\sigma_s\}, \{p_{mi}(s)\})] = \sum_m \sum_i \sum_s p_{mi}(s) \(\ln \pi_s + \ln l(m,i,s)\),		\\
	\text{s.t.~} & u_{\min} \leq u \leq u_{\max}, 	\label{u_const}\\
	 & \alphamin \leq \alpha_m, \qquad \forall m,\\
	 & \sum_s \pi_s = 1,
\end{align}		\label{mstepopt}%
\end{subequations}%
where $u_{\min}$ and $u_{\max}$ are non-negative constants.
In order to solve problem \eqref{mstepopt}, the method of Lagrange multiplier is adopted.
The Lagrangian for the problem is defined as
\begin{align}
	{\cal L}(q, k | \pi, u, \{\alpha_m\}, \{\sigma_s\}, \{p_{mi}(s)\}, \{\lambda\})
	&= \sum_m \sum_i \sum_s p_{mi}(s) \(\ln \pi_s + \ln l(m,i,s)\)		\notag\\
	&\qquad - \lambda_{\pi s} \(\sum_s \pi_s-1 \)		\notag\\
	&\qquad - \lambda_{u\min}(u_{\min}-u)-\lambda_{u\max}(u-u_{\max})		\notag\\
	&\qquad + \sum_m \lambda_{\alpha_m \min}\alpha_m
	%&\qquad - \sum_m \lambda_{\alpha_m \min}(\alpha_{\min}-\alpha_m)		%if alpha_min is non-zero
%
%	&= \sum_m \sum_i p_{mi}(\F) \(\ln \pi_\F - \ln\(2 \pi \sigma^2\) - \frac{1}{2\sigma^2} (q_{mi}-u k_{mi})^2\)	\notag\\
%	&\qquad+ \sum_m \sum_i p_{mi}(\C) \(\ln \pi_\C - \ln\(2 \pi \sigma^2\) - \frac{1}{2\sigma^2} (q_{mi}-\alpha_m + w k_{mi})^2\)	\notag\\
%	&\qquad- \lambda \(\sum_s \pi_s-1 \)
		\label{Lag}
\end{align}
where each $\lambda$ with subscript denotes a Lagrangian multiplier corresponding to each constraint in problem \eqref{mstepopt}.

By adopting the standard procedure that exploits the Karush--Kuhn--Tucker condition \citep{Karush1939kkt, Kuhn1951kkt}, this problem is analytically solved.
The solution is summarized as follows:
\begin{align}
	u=\left\{\begin{array}{ll}
		\tilde{u} & \text{if } u_{\min} < \tilde{u} < u_{\max}		\\
		u_{\min} & \text{if }\tilde{u} \leq u_{\min}		\\
		u_{\max} & \text{if }\tilde{u} \geq u_{\min},
	\end{array}\right.
\end{align}
with
\begin{align}
	\tilde{u} = \frac{\sum_m \sum_i p_{mi}(\F) q_{mi} k_{mi}/r_m^2}{\sum_m \sum_i p_{mi}(\F) k_{mi}^2/r_m^2},	\label{lag_u}
\end{align}
and
\begin{align}
	\alpha_m=\left\{\begin{array}{ll}
		\tilde{\alpha}_m & \text{if } \alphamin < \tilde{\alpha}_m		\\
		\alphamin & \text{if }\tilde{\alpha}_m \leq \alphamin,
	\end{array}\right. \qquad \forall m
\end{align}
with
\begin{align}
	\tilde{\alpha}_m = \frac{\sum_i p_{mi}(\C)(q_{mi}+wk_{mi})}{\sum_i p_{mi}(\C)},		\qquad \forall m
\end{align}
and
\begin{align}
	&\sigma_\F^2 = \frac{\sum_m\sum_i \( p_{mi}(\F)(q_{mi}-uk_{mi})^2 / r_m^2\)}{\sum_m\sum_i p_{mi}(\F)},		\\
	&\sigma_\C^2 = \frac{\sum_m\sum_i \( p_{mi}(\C)(q_{mi}-\alpha_m+wk_{mi})^2 / r_m^2\)}{\sum_m\sum_i p_{mi}(\C)},		\\
	&\pi_s = \frac{\sum_m \sum_i p_{mi}(s)}{\sum_{r \in \{\F, \C\}} \sum_m \sum_i p_{mi}(r)},	\qquad \forall s \in \{\F, \C\}.		\label{lag_pi}
\end{align}
Notice that the right-hand side of Eq.~\eqref{lag_u} is indefinite if $p_{mi}(\F)=0$ holds for all $m,i$; in other words, it is not possible to determine the free-flow speed if data are not observed from the free-flowing regime.
This corresponds to condition (i) in Corollary \ref{coro_fd_identify_triangle}.

\subsection{Summary of the algorithm}\label{sec_algo}

Input for the proposed algorithm is as follows:
\begin{itemize}
	\item Trajectories of multiple probe vehicles
	\item Jam density: $\kj$
	\item Lower and upper bounds for free-flow speed: $u_{\min}$ and $u_{\max}$
	\item Threshold for near-steady state $\theta_\text{steady}$
	\item Cleaning parameter for PTSs based on Corollary \ref{coro_fd_identify_triangle}: $\theta_\text{corr}$, $\theta_\text{std}$.
	\item Parameter sets determining shape of time--space region $\aa$: $\Delta m$, $\Delta t$, and $\phi$
	\item Initial state of the EM algorithm: $p_{mi}(s)$
	\item Threshold for the convergence of the EM algorithm: $\varepsilon$
\end{itemize}

Main output of the proposed algorithm is as follows:
\begin{itemize}
	\item Free-flow speed: $u$
	\item Backward wave speed: $w$
	\item Standard deviation of free-flowing flow from mean FD: $\sigma_\F$
	\item Standard deviation of congested flow from mean FD: $\sigma_\C$
	\item Number of vehicles between the probe pairs: $c_m$ 
\end{itemize}

The entire procedure of the proposed algorithm is summarized as follows:
\begin{description}
	\item[Step 1] Near-steady PTS set construction step
	\begin{description}
		\item[Step 1.1] Construct raw PTS set by calculating PTSs for all $m_+ \in \PP$, $m_-=m_--\Delta m$, $\Delta m \in \Delta M$, $t \in T$, $\Delta t \in \Delta T$, and $\phi \in \Phi$ based on the given probe vehicle data and definition \eqref{pts_def} and \eqref{adef}, where $\Delta M$, $\Delta T$, and $\Phi$ denote sets of given parameter values.
		\item[Step 1.2] Construct near-steady PTS set by removing specific PTSs from the raw PTS set. 
			Specifically, remove non-steady PTSs based on \equref{statidentify} with given $\theta_\text{steady}$ and non-appropriate PTSs that do not satisfy Corollary \ref{coro_fd_identify_triangle} based on given $\theta_{\text{corr}}$ and $\theta_\text{std}$.
	\end{description}
	\item[Step 2] Backward wave speed estimation step
	\begin{description}
		\item[Step 2.1] Estimate the backward wave speed $w$ by performing a linear regression on PTSs with $v\leq u_{\min}$ in the near-steady PTS set.
	\end{description}
	\item[Step 3] Free-flow speed estimation step
	\begin{description}
		\item[Step 3.1] Define the likelihood function \eqref{ELL}. 
			Give an initial state on $p_{mi}(s)$.
		\item[Step 3.2] Update $\pi_s,u,\{\alpha_m\},\{\sigma_s\}$ by executing the M-step.
		\item[Step 3.3] Update $p_{mi}(s)$ by executing the E-step.
		\item[Step 3.4] Check the convergence: if the change rate of all the parameters is lower than a given threshold $\varepsilon$, then it converges.
			If it converges, output the current parameter values and halt the algorithm.
			If not, go to Step 3.2.
	\end{description}
\end{description}

\subsection{Discussion}\label{sec_algo_disc}

The proposed method estimates mean free-flow speed and backward wave speed of a triangular FD along with other auxiliary variables. 
Two of the auxiliary variables correspond to the standard deviation of flow from the mean FD; therefore, the method also captures the stochasticity of the FD.
The method statistically estimates an FD based on data collected by multiple pairs of probe vehicles whose behavior may be heterogeneous as in general traffic.
As long as the number of probe vehicles are sufficient and they are randomly sampled, the estimation results can be considered as a representative of the whole traffic.

Although the method requires a few input parameters, their values can be determined easily (i.e., no need of fine tuning) and the estimation result would be insensitive to most parameters with the exception of $\theta_\text{steady}$ and $\kappa$.
The parameter $\theta_\text{steady}$, threshold for near-steady state identification, must be important for the proposed method, since it is relevant to the original definition of the FD.
This issue is investigated by performing a sensitivity analysis in Section \ref{sec_val_res}.
The jam density could be determined reasonably as discussed in Section \ref{sec_discussion}; moreover, even if the jam density is not known, the free-flow speed and the backward wave speed are estimated as discussed in Section \ref{sec_theory}.
The other parameters are not essential.
Any combinations of $\Delta m$, $\Delta t$, $\phi$, $\theta_\text{corr}$, and $\theta_\text{std}$ are acceptable as long as it produces sufficient number of near-steady PTSs from a given probe vehicle dataset such that the conditions (i) and (ii) of Corollary \ref{coro_fd_identify_triangle} are satisfied.
It is not necessary for the upper and lower bounds for free-flow speed, $u_{\max}$ and $u_{\min}$, to be close to the actual free-flow speed.
The upper bound $u_{\max}$ is only required to eliminate non-realistic local optima with unrealistically high $u$ (e.g., over 1000 km/h) as discussed subsequently in this section.
The lower bound $u_{\min}$ is required to determine congested traffic speed that exhibits significantly slower speed than the free-flow speed; therefore low values (e.g., 60 km/h for a highway) are acceptable.

The objective function of the problem \eqref{quasiL} is not necessarily convex nor unimodal.
Thus, local optima that does not represent an FD may exist.
For example, if the upper limit for free-flow speed \eqref{u_const} does not exist, obvious local optima are obtained at $u$ with unrealistically high speeds, such as 1000 km/h, which classifies all the PTSs belonging to congested regime.
The constrain \eqref{u_const} is useful to eliminate such local optima.

The proposed method estimates $u$ and $w$ separately although it is possible to formulate a joint estimation problem of $u$, $w$, and other auxiliary variables as a single optimization problem similar to \eqref{mstepopt} \citep{Seo2017probefd_procedia, kawasaki2017probefd}.
However, the separate estimation approach is adopted for the sake of empirical performance, because it was found that a joint estimation approach is difficult to solve stably due to numerous number of local optima, originated from fluctuations in actual traffic \citep{kawasaki2017probefd}.
Note that $u$ cannot be estimated by linear regression as in $w$; because it is not possible to extract data from free-flow regimes without knowing a precise value of $u$ in prior.

Reliability of an estimate of the proposed method can be assessed by a bootstrap method \citep[c.f.,][]{Bishop2006en}, in which the method is repeatedly executed with re-sampled data in Monte Carlo simulation manner.
Note that standard, simpler criteria such as the likelihood ratio and the Akaike information criteria cannot be applied to assess the reliability of the method with different parameter values, especially $\theta$.
This is because both of the number of model parameters and the number of samples depend on the value of $\theta$.

It must be noted that the proposed near-steady PTS extraction method is not theoretically perfect, since it only considers speed of probe vehicles.
In fact, as long as only probe vehicle data is used, it is not possible to exactly identify steady state; because behavior of non-probe vehicle does not necessarily influences probe vehicles.
Therefore, from practical perspective, it would be important to have an acceptable extraction method.
As we will see in Section \ref{sec_validation}, the proposed extraction method is acceptable in terms of its empirical performance.

The relation between the proposed method and the manual inference method proposed by Appendix in \citet{Herrera2010probe} is as follows.
They share a few of the ideas: that is, they both estimate triangular FD parameters, namely $u$ and $w$, based on probe vehicle data and a given jam density.
Meanwhile, essential differences exist.
First, the proposed method is based on the result of more general identifiability of an FD as discussed in Section \ref{sec_theory}.
Thus, the proposed method does not rely on the detection of a shockwave nor the detection of saturated free-flowing traffic (which are used by \citet{Herrera2010probe}), which are special phenomena and difficult to automatically detect by using probe vehicles.
Furthermore, the proposed method is a computable algorithm where discrepancies between the LWR theory and actual traffic are considered based on the steadiness of traffic and the noise terms.

\section{Empirical validation}\label{sec_validation}

In this section, validation results of the estimation algorithm based on real-world data are presented.

\subsection{Datasets}

We used traffic data collected at approximately 4 km length section in an inter-city highway near Tokyo, Japan (i.e., the section between Kawasaki Junction and Yokohama--Aoba Interchange in Tomei Expressway).
In this section, on-ramps, off-ramps, and merge/diverging sections are absent.
The number of lanes is three, and the speed regulation was 100 km/h.
According to the statistics, congestion occurs frequently during weekends and holidays mainly due to tourists in this section.
The data collection duration was weekends and holidays from April 2016 to September 2016, 61 days in total.

As probe vehicle data, we used data collected by Fujitsu Traffic \& Road Data Service Limited.
The data were collected by connected vehicles mainly consisting of logistic trucks and vans that were connected to the Internet for fleet management purposes.
The data included map-matched location information of each vehicles for every 1 s.
Probe vehicle data collected on a day (April 9, 2016) are illustrated in \figref{data_traject} as time--space diagrams; and it indicates that some probe vehicles slowed down due to congestion between 9 and 11 a.m.
The number of trips performed by probe vehicles during the aforementioned period was 23 103 trips.
This corresponds to mean headway time between two consecutive probe vehicles 3.8 min.
However, many of probe vehicles avoided congested time periods (this is confirmed by \figref{data_traject}); therefore, headway time in congested traffic should exceed the value.

\begin{figure}[htb]
	\centering
	\includegraphics[width=0.99\hsize]{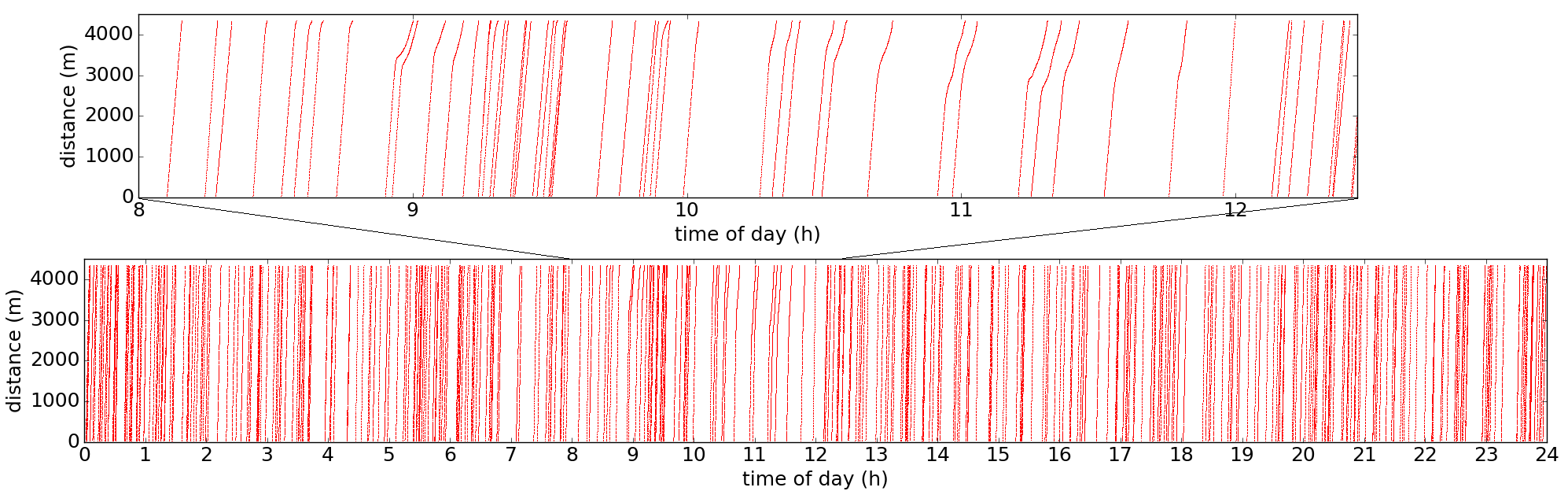}
	\caption{Probe vehicle trajectories on a day.}
	\label{data_traject}
\end{figure}

In order to validate the estimation results, loop detector data collected by Central Nippon Expressway Company Limited were used as a reference.
The detector is located in the middle of the section.
It measures vehicle type-specific flow and occupancy for every 5 min.
Density was obtained from the data by using the method proposed by \citet{cassidy1997density}.
The flow--density plot based on the detector data is shown in \figref{data_flowdensity}.
A clear triangular shape is observed; this is consistent with the existing empirical results \citep{cassidy1997density, Yan2018stationary}.
According to the plot, the flow capacity was approximately 1500 veh/h and the jam density was approximately 100 veh/km.
These values are slightly lower than those of usual highway.
This might be due to the property of drivers during the period in which most of them were potentially unexperienced drivers (i.e., sightseeing travelers in weekends) and/or possible miss detection by the loop detector.
Nevertheless, the values of the free-flow speed and the backward wave speed seems reasonable; therefore, this data can be considered as useful to validate the proposed method.

\begin{figure}[htbp]
	\centering
	\includegraphics[clip, width=0.6\hsize]{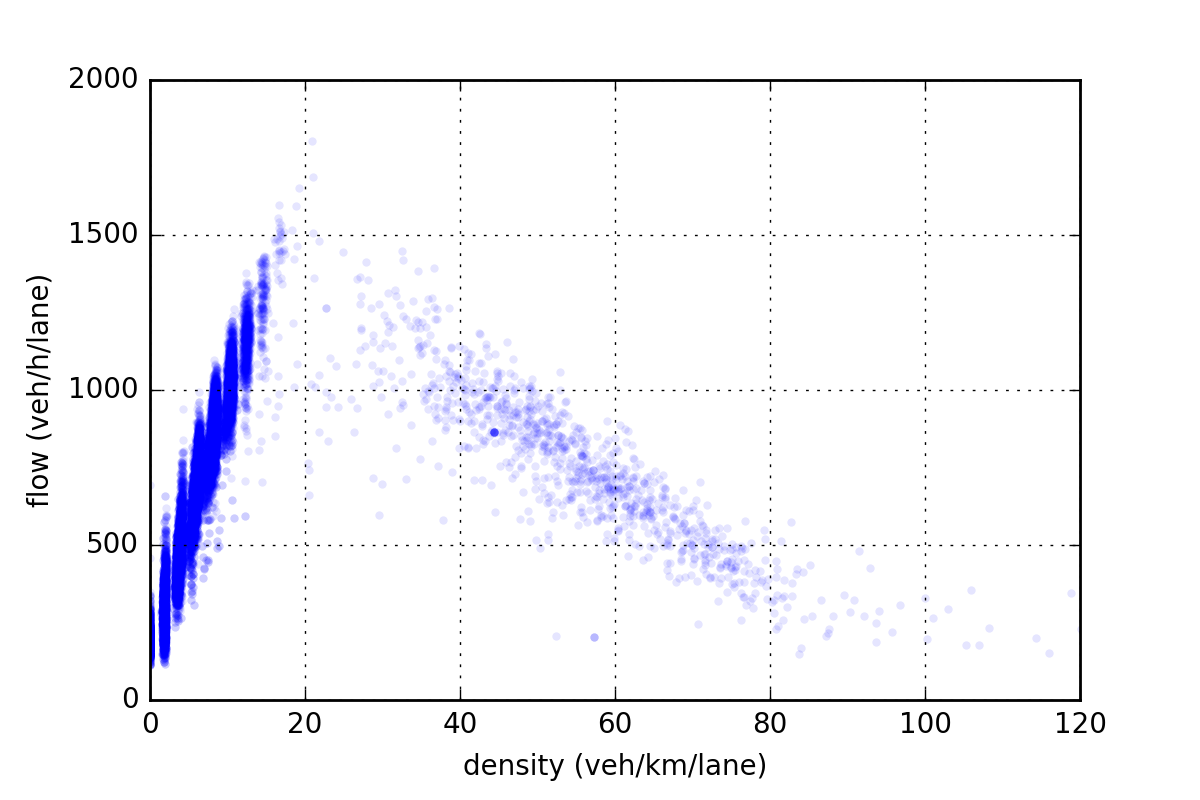}
	\caption{Reference flow--density relation based on the detector data.}
	\label{data_flowdensity}
\end{figure}

A comparison between speed measured by probe vehicles and detector for each time period is shown in \figref{data_comparison}.
Essentially, the probe vehicle data exhibits a tendency similar to that of the detector data.
However, in the high-speed regime, probe vehicle data tend to be slightly slower than the detector data.
The bias is potentially because the probe vehicles consisted of logistic trucks and vans that are typically slower than average cars.
From these results, we conclude that the probe vehicle data and detector data were sufficiently consistent in terms of speed measurement.

\begin{figure}[htbp]
	\centering
	\includegraphics[clip, width=0.6\hsize]{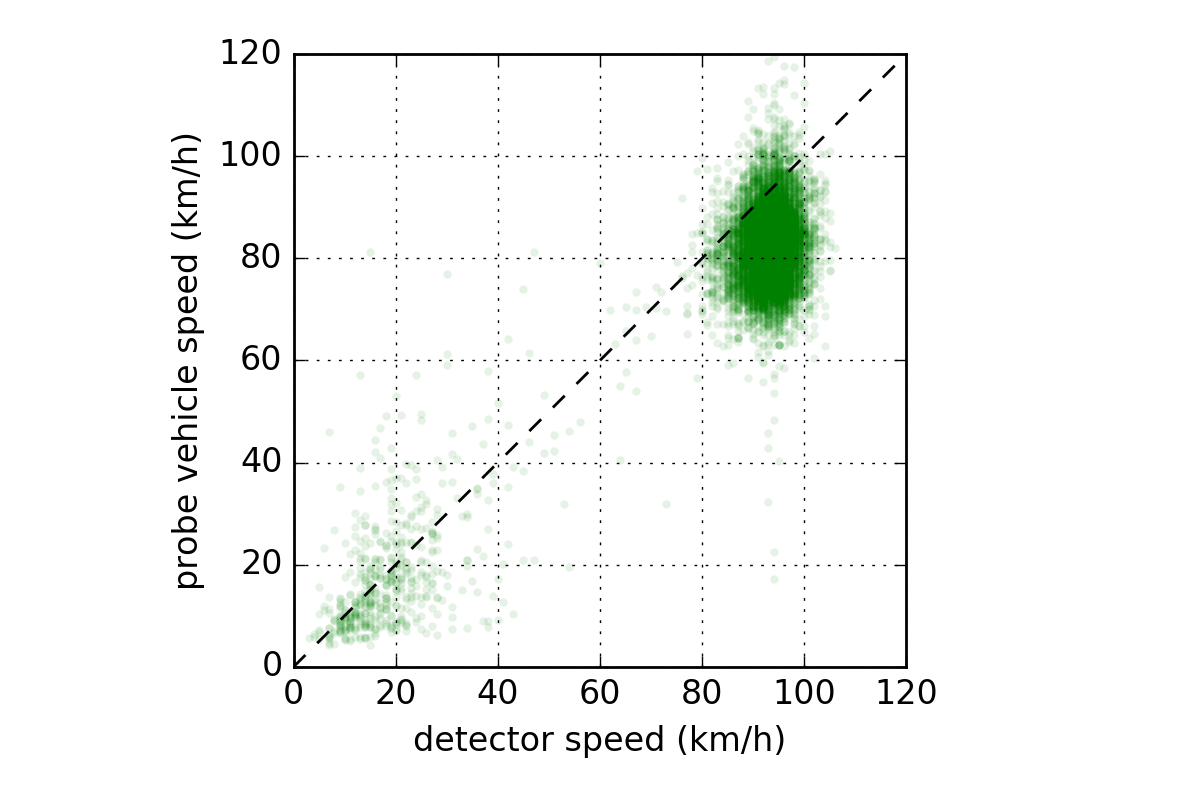}
	\caption{Comparison between speed measured by probe vehicles and detector.}
	\label{data_comparison}
\end{figure}

\subsection{Parameter setting}

The time--space regions $\aa$ were constructed based on all the combinations $(\Delta m, \Delta t, \phi)$ of the following values: $\Delta m \in \{1,2,3,4\}$, $\Delta t \in \{1,2,3,4\}$ (s), $\phi \in \{5,10,15,20,30\}$ (km/h).
The threshold for near-steady state identification, $\theta_\text{steady}$, was selected from $0.005$ to $0.3$. 
We performed a sensitivity analysis on this parameter by selecting each of them.
The cleaning parameters for PTSs were set as $\theta_\text{corr}=-0.8$ and $\theta_\text{std}=1$ (km/h).
The upper and lower bound for free-flow speed were set as $u_{\min}=60$ (km/h) and $u_{\max}=120$ (km/h), respectively, and they are considered as reasonably mild constraints for a highway with 100 km/h speed regulation.
The jam density $\kappa$ was set to $100$ veh/km based on \figref{data_flowdensity}.
The initial state of the EM algorithm is $p_{mi}(\F)=1$ and $p_{mi}(\C)=0$ for $m$ and $i$ with $v_m(\aa_{mi}) \leq u_{\min}$, and $p_{mi}(\F)=1$ and $p_{mi}(\C)=0$ for $m$ and $i$ with $v_m(\aa_{mi}) > u_{\min}$.
The convergence check parameter of the EM algorithm, $\varepsilon$, was set to $0.01$.

\subsection{Estimation results}\label{sec_val_res}

First, an estimation result under a moderate setting, $\theta_\text{steady}=0.15$, is presented.
Given the parameter setting, 1762 probe pairs produced multiple near-steady PTSs that appeared to satisfy the conditions in Corollary \ref{coro_fd_identify_triangle}, and the total number of PTSs was 437 337.
The estimated values were free-flow speed $u=81.2$ (km/h) and backward wave speed $w=17.44$ (km/h).
The values of estimates are summarized in \tabref{esti_summary} along with other parameter settings.

The estimation result with $\theta_\text{steady}=0.15$ is illustrated in \figref{esti_fd_150}.
The solid red line denotes mean estimates for $u$ and $w$, the dashed lines represent the mean plus/minus the standard deviations $\sigma_\F$ and $\sigma_\C$ of the estimates, and the blue dots represent near-steady traffic states extracted from the detector data (by using a method similar to \citet{Cassidy1998bivariate}) for reference.
The mean estimate indicates good agreement with detector data.
The estimated free-flow speed appears slightly slower than that of detector data.
This is likely due to the biased speed of probe vehicles.\footnotemark{}
Nevertheless, it is noteworthy that many of detector-measured traffic states fall within the standard deviation.
\footnotetext{
	According to \figref{data_comparison}, mean speed of probe vehicles in free-flowing regime is roughly 83 km/h (which is similar to the estimated $u$), whereas detectors measurement is roughly 95 km/h.
}%

All the estimation results with different $\theta_\text{steady}$ are summarized in \tabref{esti_summary}.
Based on these results, sensitivity on $\theta_\text{steady}$ is investigated.
With respect to the estimated $u$ and $w$, similar results were obtained in all cases with the exception of $\theta_\text{steady}=0.3$ and $\theta_\text{steady}=0.05$.
The case with $\theta_\text{steady}=0.3$ resulted in low $u$ and excessively high $w$.
This is due to the low steadiness in the PTSs filtered by $\theta_\text{steady}=0.3$; as a result, many of PTSs covered both the free-flowing and congested regime, making them inappropriate for the FD estimation.
The case with $\theta_\text{steady}=0.05$ resulted in slightly low $w$; and this suggests that the estimates are slightly unstable if the number of PTSs are small.
With respect to the estimated $\sigma_s$, a clear tendency was obtained: namely, they tended to decrease when the $\theta_\text{steady}$ decreased.
This is an expected result since low $\theta_\text{steady}$ tends to eliminate scattered PTSs in the flow--density plane that typically exhibits relatively low steadiness.

\begin{figure}[htbp]
	\centering
	\includegraphics[clip, width=0.8\hsize]{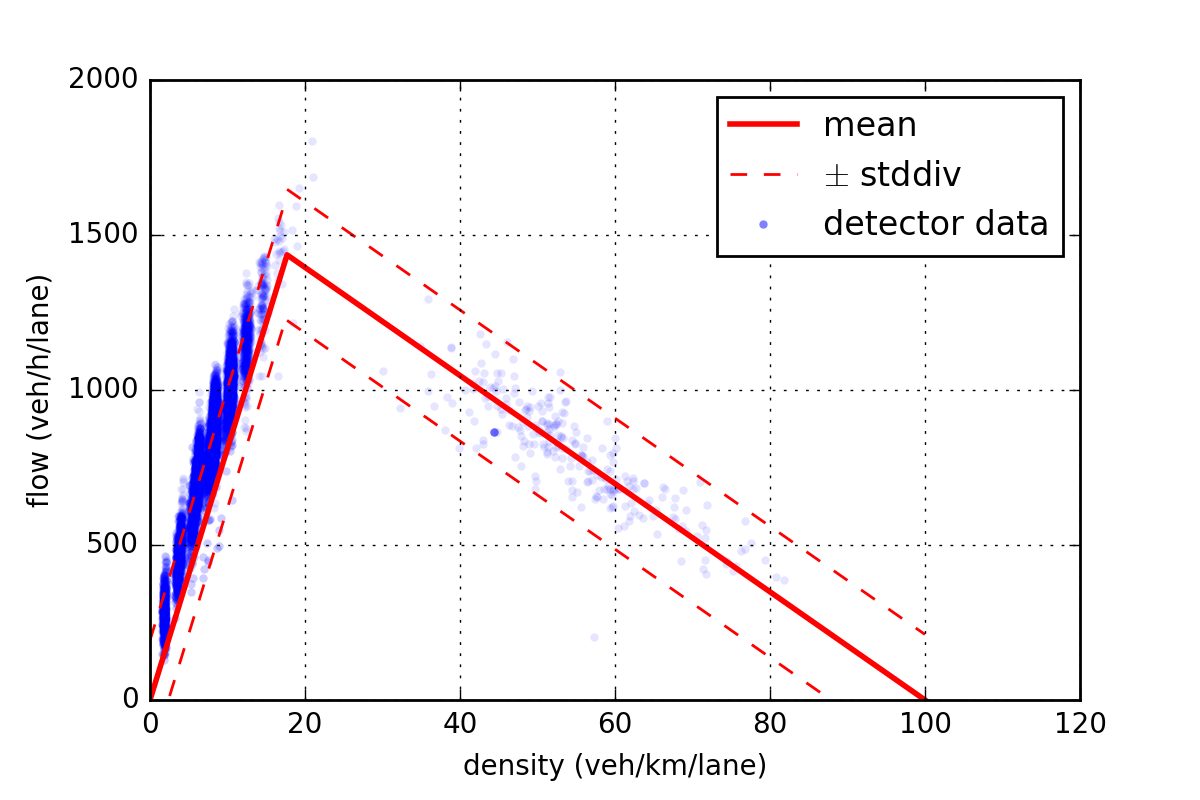}
	\caption{Estimated FD and near-steady detector data.}
	\label{esti_fd_150}
\end{figure}

\begin{table}[htbp]
	\centering
	\caption{Summary of estimation results.}
	\label{esti_summary}
	\begin{tabular}{rcrrlrrrr}
		\hline
		\multicolumn{1}{c}{Parameter} &  & \multicolumn{2}{c}{\# of near-steady data} &  & \multicolumn{4}{c}{Estimation results} \\
		\cline{1-1}\cline{3-4}\cline{6-9}
		\multicolumn{1}{c}{$\theta_\text{steady}$} &  & \multicolumn{1}{c}{Probe pairs} & \multicolumn{1}{c}{PTSs} &  & $u$ (km/h) & $w$ (km/h) & $\sigma_\F$ (veh/h) & $\sigma_\C$ (veh/h)\\
		\hline
		0.005 &  & 57 & 1 257 &  & 83.62  & 17.77  & 99.94  & 140.53  \\
		0.010 &  & 94 & 2 824 &  & 83.65  & 16.26  & 97.87  & 138.93  \\
		0.050 &  & 478 & 40 695 &  & 82.56  & 14.32  & 84.98  & 100.02  \\
		0.100 &  & 1172 & 229 892 &  & 81.94  & 16.02  & 166.82  & 181.30  \\
		0.150 &  & 1762 & 437 337 &  & 81.20  & 17.44  & 192.58  & 211.43  \\
		0.200 &  & 2344 & 654 732 &  & 80.55  & 17.63  & 197.48  & 191.81  \\
		0.300 &  & 3442 & 1 141 629 &  & 78.17  & 31.23  & 274.98  & 529.18  \\
		\hline
	\end{tabular}
\end{table}

\subsection{Discussion}

The results indicated that the proposed method estimates reasonable values of the free-flow speed and the backward wave speed, depending on the value of threshold for the steadiness $\theta_\text{steady}$.
According to \tabref{esti_summary}, the method derives fairly good estimates when $\theta_\text{steady}$ is less than $0.2$, which is a wide range for this type of variable, namely, coefficient of variation.
Conversely, as qualitatively expected, a case with excessively high $\theta_\text{steady}$ results in improper estimates because it did not properly remove non-steady traffic data.
The estimated mean FDs are slightly biased compared to the detector data, likely due to biases in probe vehicles' driving behavior.
This is one of the limitation of probe vehicle data.
Meanwhile, \ref{apx_general} suggested that the proposed methodology can estimate unbiased FD if probe vehicles are not biased.

In addition to the mean FD, the proposed method also estimates standard deviation in the FD, implying that it can capture stochasticity of the FD.
The standard deviation in free-flow regime is almost always lower than that of the congested regime; this is a reasonable result considering the nature of traffic flow.

With respect to the required amount of probe vehicle data, the results suggest that probe vehicle data with a few minutes of mean headway collected for few months are sufficient to estimate the FD.
If the penetration rate is lower than this case, the FD estimation is possible by collecting the data for a longer period.
On the contrary, if data collection duration was too short, we expect that the method may not estimate an appropriate FD, as probe vehicle data may not cover various traffic states.
An analyst needs choose an appropriate data collection duration to ensure that probe vehicle data does not violate conditions (i) and (ii) of Corollary \ref{coro_fd_identify_triangle}.
With respect to the section length required to estimate an FD, the results suggest that the length of a few kilometers is sufficient for the employed data.
It is difficult to precisely derive the minimum requirement because it depends on traffic conditions and the amount of data.
Generally, if traffic congestion occurs frequently or the amount of probe vehicle data is large, then some of the probe vehicles are likely to satisfy conditions (i) and (ii) of Corollary \ref{coro_fd_identify_triangle} in a short road section.
For this dataset, the minimum would be less than a kilometer, because some of the probe vehicles experience several traffic states at one kilometer length segment from the downstream end of this section as illustrated in \figref{data_traject}.

The optimization problem may involve unrealistic local optima as discussed in Section \ref{sec_algo_disc}.
However, in the empirical analysis, the proposed method always yielded solution with physically reasonable results with the exception of the excessively high $\theta_\text{steady}$ case.
This is potentially due to the separated estimation approach and the constraint on free-flow speed that we employed.

\section{Conclusion}\label{sec_conclusion}

A theory and an algorithm for fundamental diagram (FD) estimation by using trajectories of probe vehicles is presented.
They estimate FD parameters (i.e., free-flow speed and backward wave speed of a triangular FD) based on the jam density and time-continuous trajectories of randomly sampled vehicles.
The algorithm is developed to estimate a reasonable FD from actual noisy traffic data.
The algorithm was empirically and quantitatively validated based on real-world probe vehicle data on a highway.
The results suggested that the algorithm accurately and robustly estimates triangular FD parameters and captures the stochasticity of the FD.

The following future research directions are considerable.
First, it is important to develop a method that considers spatially heterogeneous FDs such that a bottleneck is endogenously detected.
This is possible by extending the method to estimate location-dependent $u$ and $w$ and global $c_m$.
Second, sophistication of the near-steady PTS extraction method is valuable.
For example, consideration of duration of near-steady state would be useful.
However, too strict criteria will results in the shortage of data; thus, careful design is required.
Third, extension to other FD functional forms is desirable.
\ref{apx_general} presents preliminary results.
Finally, it is considerable to incorporate of a jam density inference method based on remote sensing by using satellites (c.f., Section \ref{sec_discussion}), in order to make the method independent of exogenous knowledge.
For this purpose, the ``small satellites'' that attracted significant attention recently \citep{sandau2010satelliteb} will be useful owing to their flexible and low-cost sensing capability.
The latter direction is now being investigated by the authors \citep{Seo2018satellite_procedia}.

%%%%%%%%%%%%%%%%%%%%%%%%%%%%%%%%%%%%%%%%%%%%%%%%%%%%%%%%%%%%%%%%%%%%%%%%%%%%%%%%%%%%%%%%%%%%%%%%%%%%%%%%%%%%%%%%%%%%%%%%%%%%%%%%%
%%%%%%%%%%%%%%%%%%%%%%%%%%%%%%%%%%%%%%%%%%%%%%%%%%%%%%%%%%%%%%%%%%%%%%%%%%%%%%%%%%%%%%%%%%%%%%%%%%%%%%%%%%%%%%%%%%%%%%%%%%%%%%%%%
%%%%%%%%%%%%%%%%%%%%%%%%%%%%%%%%%%%%%%%%%%%%%%%%%%%%%%%%%%%%%%%%%%%%%%%%%%%%%%%%%%%%%%%%%%%%%%%%%%%%%%%%%%%%%%%%%%%%%%%%%%%%%%%%%

\section*{Acknowledgments}

The probe vehicle data were provided by Fujitsu Traffic \& Road Data Service Limited and the detector data were provided by Central Nippon Expressway Company Limited.
The study was financially supported by the Japan Society for the Promotion of Science (KAKENHI Grant-in-Aid for Young Scientists (B) 16K18164 and Scientific Research (S) 26220906).
The authors would like to express their sincere appreciation for the help.

\appendix
\section{List of abbreviations and notations}\label{sec_notation}

Important abbreviations and notations that are used throughout this paper are listed below.

{\footnotesize
\begin{longtable}{ll}
	\hline
	FD & 	Fundamental diagram		\\
	LWR & 	Lighthill--Whitham--Richards		\\
	FIFO & 	First-in first-out		\\
	PTS & 	Probe traffic state		\\
	PFD & 	Probe fundamental diagram		\\
	$t$ &	 Time		\\
	$x$ &	 Location		\\
	$q$ &	 Flow		\\
	$k$ &	 Density		\\
	$v$ &	 Speed		\\
	$d_n(\AA)$	& Distance traveled by vehicle $n$ in time--space region $\AA$	\\
	$t_n(\AA)$	& Time spent by vehicle $n$ in region $\AA$	\\
	$S$ &	 Traffic state. $S \equiv (q,k)$		\\
	$Q$ &	 Density-to-flow FD function		\\
%	$Q_j$ &	 Function of $j$-th piecewise differentiable part of an FD		\\
%	$g$	&	 Number of parameters of an FD		\\
%	$g_j$ & Number of parameters of $j$-th piecewise differentiable part of an FD		\\
%	$h$	 &	 Number of piecewise differentiable parts of an FD		\\
	$\kj$ &	 Jam density		\\
	$u$ &	 Free-flow speed		\\
	$w$ &	 Backward wave speed		\\
	$\alpha$ &	 $y$-intercept of the congestion part of flow--density triangular FD. $\alpha \equiv w\kj$		\\
	$m$ &	 Probe pair	(i.e., pair of two probe vehicles)		\\
	$m_-$ &	 Preceding probe vehicle in pair $m$	\\
	$m_+$ &	 Succeeding probe vehicle in pair $m$	\\
	$c_m$ &	 Number of vehicles between probe vehicles in pair $m$ plus one		\\
	$\phi$ & A parameter (angle) used to define a shape of region $\aa$		\\
	$\Delta t$ & A parameter (width) used to define a shape of region $\aa$		\\
	$\alpha_m$ &	 $y$-intercept of the congestion part of flow--density triangular PFD of pair $m$. $\alpha_m \equiv w\kj/c_m$		\\
	$\theta_{\text{steady}}$ &	 Threshold for near-steady traffic detection		\\
	$i$ &	 Index for datum		\\
	$\aa_{mi}$ &	$i$-th time--space region between probe vehicles in pair $m$		\\
	$q_m(\aa_{mi})$ &	 Probe-flow of pair $m$ in region $\aa_{mi}$		\\
	$k_m(\aa_{mi})$ &	 Probe-density of pair $m$ in region $\aa_{mi}$		\\
	$v_m(\aa_{mi})$ &	 Probe-speed of pair $m$ in region $\aa_{mi}$		\\
	$S_m(\aa_{mi})$ &	 PTS of pair $m$ in region $\aa_{mi}$. $S_m(\aa_{mi}) \equiv (q_m(\aa_{mi}), k_m(\aa_{mi}))$		\\
	$q_{mi}$ &	 Probe-flow of datum $i$ of pair $m$. $q_{mi} \equiv q_m(\aa_{mi})$		\\
	$k_{mi}$ &	 Probe-density of datum $i$ of pair $m$. $k_{mi} \equiv k_m(\aa_{mi})$		\\
	$S_{mi}$ &	 PTS of datum $i$ of pair $m$. $S_{mi} \equiv (q_{mi}, k_{mi}) = S_m(\aa_{mi})$		\\
	$s$ &	 Variable representing regime of a traffic state. $s \in \{\F, \C\}$		\\
	$\F$ &	 Free-flowing regime		\\
	$\C$ &	 Congested regime		\\
	$\pi_s$ &	 Contribution of state $s$ in the mixed distribution		\\
	$z_{mis}$ &	 1 if $S_{mi}$ belongs to regime $s$, 0 otherwise		\\
	$p_{mi}(s)$ &	 Expectation of $z_{mis}$		\\
	\hline
\end{longtable}
}

\section{Cases with general continuous FDs}\label{apx_general}
\setcounter{figure}{0}

In this appendix, the applicability of the proposed methodology to FDs with general continuous functional forms is presented.

\subsection{Identifiability}\label{apx_formulation}

We overwrite an assumption in Section \ref{sec_assumption} to make the methodology more general.
Hereafter, the FD is assumed to be continuous.

Following the same idea described in Section \ref{sec_theory}, we have following proposition on the identifiability of continuous, differentiable FDs based on probe vehicle data:
\begin{proposition}\label{pro_fd_identify_gen_differentiable}
	A continuous, differentiable FD whose number of parameters is $b$ is identifiable if (i) $b$ or more different steady PTS data are available for a probe pair and (ii) the value of the jam density is known.
\end{proposition}
\begin{proof}
	The corresponding PFD of the probe pair is identifiable from the said PTS data; this is evident from the identifiability condition of continuous, differentiable functions.
	By comparing the jam density of the PFD and that of the FD, the value of the FD parameters can be derived.
\end{proof}

\subsection{Algorithm}\label{apx_method}

Now we consider a computational algorithm of FD estimation based on potentially noisy probe vehicle data as in Section \ref{sec_method}.
Following the same idea described in Section \ref{sec_method}, an FD estimation algorithm is constructed as follows:
\begin{description}
	\item[Step 1] Near-steady PTS set construction step. Construct a near-steady PTS set by Step 1.1 of the algorithm shown in Section \ref{sec_algo}.
	\item[Step 2] FD parameters estimation step. Solve the following residual minimization problem:
	\begin{subequations}
	\begin{align}
		\min_{\tthh, \{c_m\}}& \sum_{m \in M} \sum_{i \in I_m} \(c_m q_{mi}-Q(c_m k_{mi}; \tthh)\)^2,	\label{min_diff}\\
		\mathrm{s.t.~} & c_m \geq 0, \qquad \forall m,		\label{cm_heuristic}\\
		 & \mathrm{physical~constraints~on~} \tthh,
	\end{align}\label{apndx_prob}%
	\end{subequations}
		where $\tthh$ denotes the FD parameters vector, $(q_{mi}, k_{mi})$ denotes a near-steady PTS datum $i$ of probe pair $m$, and $c_m$ denotes number of vehicles between a probe pair $m$ plus one.
		Note that we use Eq.~\eqref{cm_heuristic} as a heuristic constraint on $c_m$, although true physical constraint is $c_m \geq 1$.
		This is because of a certain property of the optimization problem: this setting greatly increases the opportunity of finding physically reasonable optima by gradient methods.
		For the reason, see \citet{Seo2017probefd_procedia}.
\end{description}

Although the expression of this algorithm is simpler than the proposed method for the triangular FD case in Section \ref{sec_method}, the former may be more difficult to solve, as the latter leverages the piecewise linearity of the triangular FD.
Development of an efficient algorithm for general FDs is out of scope of this appendix.
Instead, in this appendix, problem \eqref{apndx_prob} is solved by a quasi-Newton method with multiple initial values, which is general but potentially computationally costly.

\subsection{Simulation-based evaluation}\label{apx_evaluation}

In order to evaluate the method, synthetic data generated by traffic simulation was used.
First, a functional form and parameters for an underlying FD was assumed: we used either of a triangular FD, Greenshield's FD ($Q_{GS}(k)=uk(1-k/\kappa)$), or Greenberg's FD ($Q_{GB}(k)=v_c k \ln\(\kappa/k\)$, where $v_c$ is a parameter) for this illustration purpose.
The values of parameters were $u=80$ (km/h), $w=15$ (km/h), $\kappa=200$ (veh/km), and $v_c=30$ (km/h).
Then, traffic states in 1 km homogeneous road section with various demand and supply conditions were generated by the underlying FD.
A flow-dependent perturbation was added to the generated traffic states to simulate system and observation noises.
Probe vehicles that record their location in every 1 seconds were randomly sampled with penetration rate of 5{\%} from the generated ground truth traffic, so that 50 probe pairs were constructed.
Finally, the FD estimation problem \eqref{apndx_prob} was solved using L-BFGS-B algorithm \citep{byrd1995lbfgsb} with 30 random initial values.

\begin{figure}[tb]
	\centering
	\subfloat[Triangular FD. Estimated values: $\hat{u}=80.2$ (km/h) and $\hat{w}=14.9$ (km/h)]{\includegraphics[width=0.3\hsize]{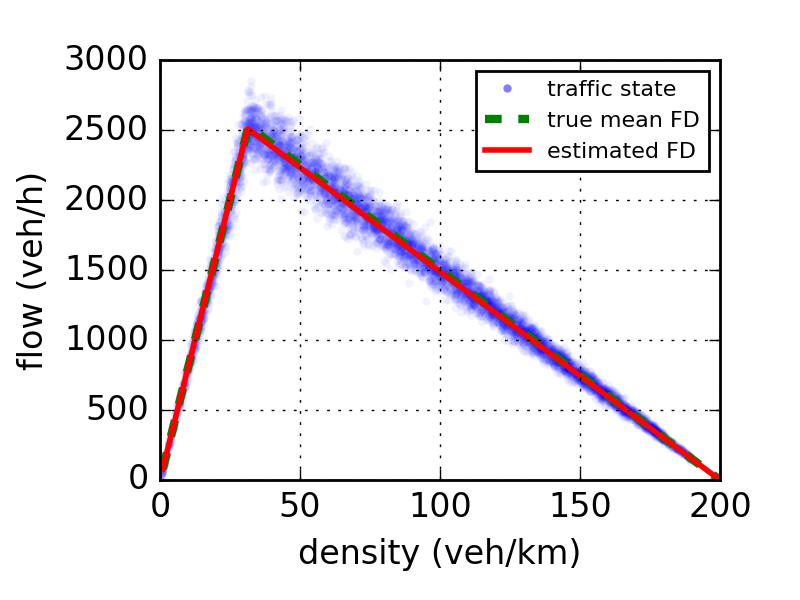}\label{resTR}}~
	\subfloat[Greenshield's FD. Estimated values: $\hat{u}=79.8$ (km/h)]{\includegraphics[width=0.3\hsize]{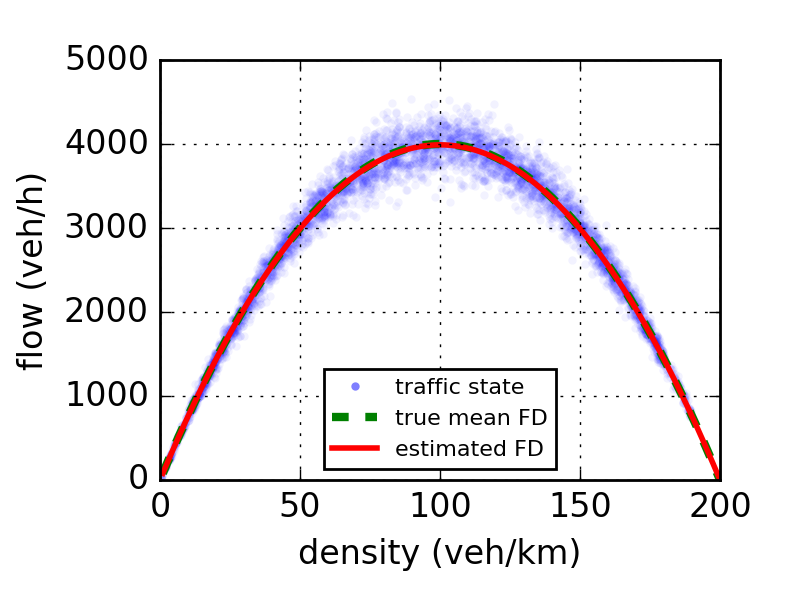}\label{resGS}}~
	\subfloat[Greenberg's FD. Estimated values: $\hat{v}_c=29.9$ (km/h)]{\includegraphics[width=0.3\hsize]{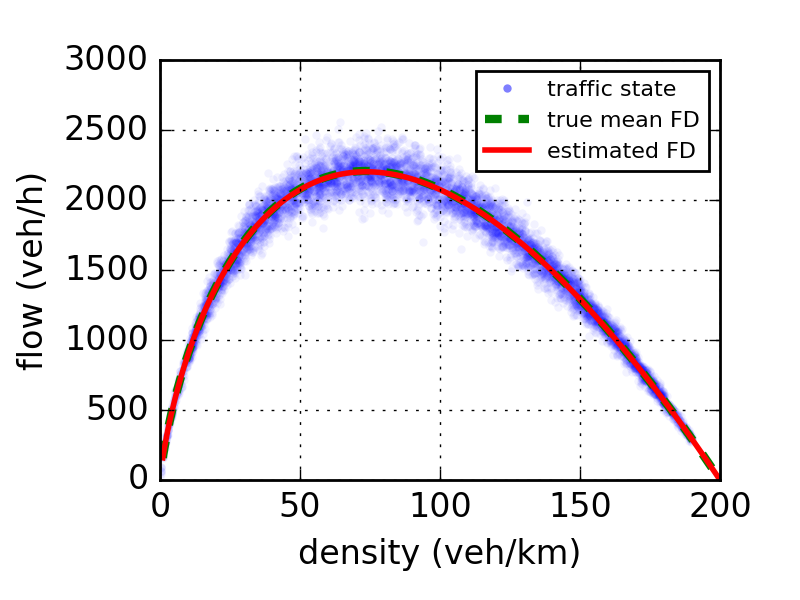}\label{resGB}}
	\caption{Simulation-based FD estimation results.}
	\label{appendix_res}
\end{figure}

Estimation results are summarized in Fig.~\ref{appendix_res}.
It is clear that estimated FDs almost perfectly coincide with the corresponding true mean FDs.
From these results, we conclude that the proposed methodology can estimate some continuous FDs in addition to the triangular FDs.
Additionally, the proposed methodology can estimate an unbiased triangular FD if probe vehicle data was not biased.
These findings complement the findings of the case study with real probe vehicle data in Section \ref{sec_validation}.

%\section*{References}

{\small
%\bibliography{ref_slim}

\bibliographystyle{elsarticle-harv}
}

\end{document}